\documentclass[11pt]{article}
\usepackage[margin=3 cm]{geometry}
\usepackage{graphicx,amssymb,amsmath,amsthm, latexsym} 
\usepackage{algorithmic}
\usepackage{amsfonts}
\usepackage{hyperref,color,xcolor}
\usepackage{enumerate}
\usepackage{comment}
\usepackage{mathtools}
\usepackage{verbatim} 

\usepackage{enumitem}

\newtheorem{theorem}{Theorem}[section]

\newtheorem{lemma}[theorem]{Lemma}
\newtheorem{proposition}[theorem]{Proposition}

\theoremstyle{definition}
\newtheorem{definition}{Definition}[section]

    \newtheoremstyle{TheoremNum}
        {\topsep}{\topsep}              
        {\itshape}                      
        {}                              
        {\bfseries}                     
        {.}                             
        { }                             
        {\thmname{#1}\thmnote{ \bfseries #3}}
    \theoremstyle{TheoremNum}
    \newtheorem{thmn}{Theorem}
		
		    \newtheoremstyle{TheoremNum}
        {\topsep}{\topsep}              
        {\itshape}                      
        {}                              
        {\bfseries}                     
        {.}                             
        { }                             
        {\thmname{#1}\thmnote{ \bfseries #3}}
    \theoremstyle{TheoremNum}

\newcommand{\RR}{\mathbb{R}}
\newcommand{\NN}{\mathbb{N}}

\newcommand{\sgn}{\operatorname{sgn}}

\newcommand{\argmax}{\operatorname{argmax}}

\newcommand{\Payoff}{\operatorname{Payoff}}

\newcommand{\err}{\operatorname{err}}

\def\al{\alpha}
\def\E{\mathbf{E}}

\title{Incentive Compatible Active Learning}
\author{Federico Echenique\thanks{%
California Institute of Technology,
{\sl fede@hss.caltech.edu}
} \and Siddharth Prasad\thanks{%
Carnegie Mellon University,
{\sl sprasad2@cs.cmu.edu}
}}

\begin{document}
\maketitle

\begin{abstract} We consider active learning under incentive compatibility constraints. The main application of our results is to economic experiments, in which a learner seeks to infer the parameters of a subject's preferences: for example their attitudes towards risk, or their beliefs over uncertain events. By cleverly adapting the experimental design, one can save on the time spent by subjects in the laboratory, or maximize the information obtained from each subject in a given laboratory session; but the resulting adaptive design raises complications due to incentive compatibility. A subject in the lab may answer questions strategically, and not truthfully, so as to steer subsequent questions in a profitable direction.

We analyze two standard economic problems: inference of preferences over risk from multiple price lists, and belief elicitation in experiments on choice over uncertainty. In the first setting, we tune a simple and fast learning algorithm to retain certain incentive compatibility properties. In the second setting, we provide an incentive compatible learning algorithm based on scoring rules with query complexity that differs from obvious methods of achieving fast learning rates only by subpolynomial factors. Thus, for these areas of application, incentive compatibility may be achieved without paying a large sample complexity price.
\end{abstract}

\section{Introduction}
\label{sec:typesetting-summary}
We study active learning under incentive compatibility constraints. Consider a learner: Alice, who seeks to elicit the parameters governing the behavior of a human subject: Bob. The chief application of our paper is to the design of laboratory experiments in economics. In such applications, Alice is an experimenter observing choices made by Bob in her laboratory. The active learning paradigm seeks to save on the number of questions posed by Alice by making the formulation of each question dependent on Bob's answers to previous questions \cite{balcan2009agnostic, dasgupta2011two}. Now, Bob may misrepresent his answers to some of Alice's questions so as to guide Alice's line of questioning in a direction that he can benefit from.

Our setting differs from standard applications of active learning in computer science, in that data are labeled by a self-interested human agent (in our story, Bob). Computer scientists have thought of active learning as applied to, for example, combinatorial chemistry, or image detection. A learner then makes queries that are always truthfully answered. In economic settings, in contrast, one must recognize the role of incentives. 

The existing literature on applications of passive learning to preference elicitation (see for example \cite{beigman2006learning,kalai2003learnability,basu2018learnability,chase2018learning}) does not have to deal with agents' incentives to manipulate the learning mechanism, but active learning does, because an agent who understands the learner's algorithm may answer strategically early on in the experiment so as to influence the questions he faces later in the experiment. 

We should emphasize that experimental orthodoxy in economics requires that subjects (such as Bob) know as much as possible about the experimental design. No deception is allowed in economic experiments. In addition, subjects' participation is almost universally incentivized: Bob gets a payoff that depends on his answers to Alice's questions. Our model relates to a long-standing interest among economists for adaptive experimental design, see \cite{el1993bayesian,ray2012bayesian,chapman2018loss,imai2016estimating}.

Consider a concrete example. Bob has a utility function $x^\sigma$ over money, so that if he faces a random amount of money $X$, his expected utility is $\E[X^\sigma]$. In other words, Bob has a utility of the ``constant relative risk aversion'' (CRRA) form, and Alice wants to learn the value of the parameter $\sigma$ --Bob's relative risk aversion coefficient.\footnote{The coefficient $\sigma$ captures Bob's willingness to assume risk. It is a parameter that economic experiments very often seek to measure, even when the experiment is ostensibly about a totally different question. Economic experimentalists  want to understand the relation between risk and their general experimental findings, so they include risk elicitation as part of the design.} A standard procedure for estimating $\sigma$ is a {\em multiple price list}.\footnote{Multiple price lists are a very common experimental design, first used by \cite{binswanger1981attitudes}, and popularized by \cite{holt2002risk} as a method to estimate $\sigma$, as described here.} 

In a multiple-price list (MPL), Alice successively asks Bob to choose between a sure payoff of $x$ dollars and a fixed lottery $L$, for example a lottery that flips a fair coin and pays $0$ dollars if the coin turns up Heads, and $1$ dollar if it turns up Tails. Alice would first ask Bob to choose between a very small amount $x$ (almost zero) and $L$. Then Alice would raise $x$ a little and ask Bob to choose again. The procedure is repeated, each time increasing the amount $x$, until reaching a number equal to, or close to 1. At some value $x$, Bob would switch from preferring the lottery to preferring the fixed amount of money. Then Alice would solve the equation 
\begin{equation}\label{eq:intro}
  x^\sigma = (1/2)0^\sigma+(1/2)1^\sigma = (1/2) 
  \end{equation} 
to find the value of $\sigma$. Now, it is important to explain how the experiment is incentivized: When the experiment is over, Alice will actually implement one of the choices made by Bob. Conventional experimental methodology dictates \cite{azrieli2018incentives} that she chooses one of the questions at random and implements it. 

A proponent of active learning will immediately remark that the MPL  design asks too many questions. Alice only needs to know the value of $x$ at which Bob is indifferent between $x$ and the lottery $L$. We can thus imagine an adaptive design, where Alice raises $x$ until Bob switches from $L$ to $x$, and stops the experiment when that happens. This design will result in strictly fewer questions than the passive (supervised) learning design.

Bob, however, understands that Alice stops raising $x$ when he declares indifference to $L$. So he will manipulate Alice into offering him values of $x$ {\em beyond} what he truly views as indifferent to $L$. Specifically, suppose that Alice raises $x$ continuously (this is a simplifying assumption; see Section~\ref{sec:MPL} for a realistic version of this design), and that if Bob declares indifference at $x$ then the last question is implemented with probability $p(x)\in (0,1)$. The function $p$ is strictly decreasing since reporting a larger value of $x$ increases the probability that a question for which Bob preferred $L$ will be implemented.

Bob's payoff from stopping at $x$ is $\pi(x;\sigma)=p(x)x^\sigma+(1-p(x))(1/2)$ because with probability $p(x)$ the last question gets implemented, so he gets the sure amount $x$, and with the complementary probability one of the other questions is implemented and Bob gets his preference for those questions, namely the lottery $L$. The expected utility of $L$ is $1/2$.

Then it is clear that Bob would like to stop at an $x$ that is strictly greater than the value at which he would truly be indifferent to $L$, the value that solves~\eqref{eq:intro}. If he stops at the true value of $x$, he gets for sure something that he values as much as $L$ (either $L$ or the amount $x$ that he values exactly as $L$). By stopping at a strictly greater $x$, he has a shot at getting a value of $x$ that he prefers over $L$.

The situation is, however, far from hopeless. Bob's optimal value of $x$ is strictly increasing in $\sigma$.\footnote{If $\pi(x;\sigma)=p(x)x^\sigma + (1-p(x))(1/2)$ and we assume that $p$ is smooth, then $\partial^2 \pi(x;\sigma)/\partial\sigma\partial x = p'(x)x^\sigma \ln(x) + (\sigma+1)p(x)x^{\sigma-1}>0$ as $x\in (0,1]$ and $p$ is decreasing. So $\pi$ is strictly supermodular. Hence the optimal $x$ is increasing in $\sigma$.} Alice can then undo Bob's strategic choice of $x$ and back out the true value of $\sigma$. (Alice's approach is common in applied econometrics, often called the ``structural'' method.)

In this paper, we prove general possibility results, to illustrate that there are many situations where active learning is consistent with incentive compatibility. In Section~\ref{sec:MPL} we shall present a formal model of multiple price lists, and show that it is possible to learn while satisfying incentive compatibility. In Section~\ref{sec:general} we discuss incentive issues in active learning in a more general sense. We present a formal notion of incentive compatible active learning in a general preference elicitation environment, and provide characterizations of the complexity of incentive compatible learning in certain ``nice'' environments.

A recent and growing body of work studies the problem of inferring models of economic choice from a learning theoretic perspective \cite{kalai2003learnability,beigman2006learning, zadimoghaddam2012efficiently,  balcan2014learning, basu2018learnability, chase2018learning, basu2019learnability}. The learnability of preference relations has also received very recent attention \cite{basu2018learnability, chase2018learning}. Our investigation takes a different angle: in attempting to model an experimental situation where subjects are asked to make choices in an interactive manner, via, e.g., a computer program, or in person, we allow the analyst complete control over the learning data. In the active learning literature, this framework is known as the \emph{membership queries} model. There is also an ongoing line of work that considers learning problems when the data provider is strategic \cite{dekel2010incentive, abernethy2015low, liu2017machine, chen2018strategyproof}. Finally, the recent work of \cite{hardt2016strategic} studies a model where an agent may (at a cost) manipulate the input to a classification algorithm. 

The membership query model closely captures an adaptive economic experiment, while in this context the more traditional learning/active learning models (e.g.\ PAC learning, stream-based active learning) seem to place unnecessary restrictions on how the analyst learns. This notion is briefly discussed in \cite{chase2018learning}, where classical learning theoretic approaches appear to give much weaker complexity guarantees than the membership queries model in learning time-dependent discounted utility preferences. \emph{For the remainder of this paper, whenever we use the phrase ``active learning'', we refer to the membership query setting -- all other forms of learning can be viewed as a special case of membership queries.} 

The other important component in modelling an economic experiment is a payment to the agent after the experiment has concluded. Experiments in economics are always {\em incentivized}, meaning that there are actual material consequences to subjects' decisions in the lab. Subjects are paid for their decisions in the experiment.  We incorporate this incentive payment into the execution of the algorithm by which the analyst chooses questions -- the analyst implements the outcome chosen by the agent in the final round of the interaction. Thus, rather than treating the payment scheme as a separate problem, we use it to demand a certain level of robustness from our learning algorithms. As we demonstrate, this precludes the analyst from running naive learning algorithms that, despite achieving good query complexities, allow the agent to strategically and dishonestly answer questions to get offered higher payoff outcomes.

Finally, the framework we introduce engenders the following natural question: is there a combinatorial measure of complexity, akin to VC dimension for PAC learning concept classes, that precisely captures the complexity of incentive compatible learning in preference environments? Our results examine certain sufficient conditions for incentive compatible learning, a potential first step towards better understanding this new and interesting learning model.

\subsubsection*{Summary of results}

We begin by discussing incentive issues in a very common experimental paradigm, that of convex budgets. We present an example to the effect that incentive problem are present and can be critical. Then we turn to the \emph{Multiple Price Lists} (MPL), another very common experimental design used to infer agents' attitudes towards risk. In MPL experiments, an agent is asked to choose between receiving various deterministic monetary amounts and participating in a lottery. The goal of the analyst is to elicit the agent's \emph{certainty equivalent}, i.e.\ the deterministic quantity at which the agent values the lottery (in our previous discussion, the certainty equivalent is the quantity that solves Equation~\eqref{eq:intro}). We analyze a simple sequential search mechanism that is used in practice -- start from the lowest possible deterministic amount and keep increasing the offer until the agent prefers it to the lottery. The analyst pays the agent by implementing the agent's decision on a randomly selected question that was asked. We show that while this mechanism is not incentive compatible, under relatively benign assumptions it satisfies a one-to-one condition where the analyst can accurately infer the agent's true certainty equivalent after learning the agent's reported certainty equivalent. We then show how a modified payment scheme that only depends on the final decision of the agent allows the analyst to do a binary search and retain incentive compatibility, giving a mechanism for learning the certainty equivalent of a strategic agent to within an error of $\varepsilon$ using $O(\log 1/\varepsilon)$ questions. 

We then turn to an abstract model of learning preference parameters/types. The idea is, as in the MPL, to induce incentive compatibility by incentivizing the payment from the last question asked of the agent. To this end, we coin a formal notion of incentive compatible (IC) learnability. A learning algorithm is simply an adaptive procedure that at each step asks the agent to choose between two outcomes. Informally, the \emph{IC learning complexity} of an algorithm is the number of rounds required to both

\begin{enumerate}
    \item Accurately learn (with high probability) the agent's type with respect to some specified metric on the type space.
    \item Ensure that (with high probability) the payment mechanism of implementing the agent's choice on the final question cannot be strategically manipulated to yield a significant payoff gain over answering questions truthfully.
\end{enumerate}

A simple structural condition allows a strong notion of incentive compatibility to be achieved via a deterministic exhaustive search (truthful reporting is the agent's unique best response), and we give examples of commonly studied economic preference models that fit our condition. We demonstrate that a large class of preference relations over Euclidean space -- those exhibiting strict convexity under a condition which we call \emph{hyperplane uniqueness} (detailed in Section~\ref{sec:general}) -- can be learned in an incentive compatible manner. 

\begin{theorem}[informal]\label{thm:strictly_convex}
Let $\Theta$ be a type space such that the preferences induced by each $\theta\in\Theta$ are continuous, strictly convex, and satisfy hyperplane uniqueness. Then, $\Theta$ is IC learnable, under a suitably chosen metric.
\end{theorem}

However, this strong notion of incentive compatibility comes at a cost -- the associated IC learning complexity can be massive (exponential in the preference parameters). In the abstract setting of preferences over outcomes, it is unclear how to obtain a tangible improvement in this complexity (even with randomization), and specifically it would appear that the problem parameters (e.g. the outcome space, the set of possible agent types) require much more structure for any sort of improvement. 

We then analyze the specific setting of learning the beliefs of an expected utility agent, where we have the required structure. Here, an agent holds a belief represented by a distribution $\alpha\in\Delta_n$ (there are $n$ uncertain states of the world, and $\alpha_i$ is the probability with which the agent believes state $i$ will occur), and is asked to make choices between vectors of rewards $x\in\RR^n$, where the utility an agent of type $\alpha$ enjoys from $x$ is simply $\alpha.x$. We first observe that naive learning algorithms can vastly beat the learning complexity of the general preference framework, but fail to be incentive compatible. Our main result is an incentive compatible learning algorithm for eliciting an agent's beliefs that significantly improves upon the complexity in the general framework, and only differs from the fast naive learning algorithms by subpolynomial factors.

\begin{theorem}\label{thm:eu}
There is an algorithm for learning the belief of an expected utility agent that when run for $$O\left(n^{3/2}\log n\max\left(\log\frac{n}{\varepsilon}, \log\frac{1}{\tau}\right)\right)$$ rounds (with high probability) cannot be manipulated to yield more than a $\tau$ increase in payoff, and learns a truthful agent's belief to within total variation distance of $\varepsilon$. \footnote{Typical supervised learning bounds have a logarithmic dependence on the confidence parameter $\frac{1}{\delta}$, and so for the sake of brevity we omit terms depending on $\delta$ in our complexity bounds.}
\end{theorem}

Our algorithm is built upon disagreement based active learning methods that provide learning guarantees, and employs the spherical scoring rule to ensure incentive compatibility properties.

\section{Example: Convex budgets}
We present a simple example to illustrate how incentive issues can prevent a very popular experimental design from being implementable in an active learning setting.

Consider an experiment on choice under uncertainty, with an adaptive ``convex budgets'' design. Such designs are ubiquitous in experimental economics: see \cite{andreoni2002,choi2007,ahn2014,friedman2018varieties,andreoni2003,augenblick2013} among (many) others. Convex budgets is very popular as a design because it parallels the most basic model in economic theory, the model of consumer choice.\footnote{Consumer choice is probably the first model a student of economics is ever exposed to. It captures optimal choice from an economic budget sets, defined from linear prices and a maximum expenditure level.}

Bob, a subject in the lab, has expected utility preferences. Specifically, suppose that the experiment involves two possible states of the world, and that Bob chooses among vectors $x=(x_1,x_2)\in\RR^2_+$. If Bob chooses the vector $(x_1,x_2)$ and the state of the world turns out to be $i$, then he is paid $x_i$. Bob believes that the state of the world $i$ occurs with probability $\al_i$, so his expected utility from choosing $x$ is $\al_1 x_1 + \al_2 x_2$ (we assume for simplicity that Bob is risk-neutral).

The experiment seeks to learn the subjects' beliefs $\al$ with a design that has Bob choosing \[ x\in B(p,I) = \{y\in\RR^2_+ : p\cdot y\leq I \},\] at prices $p\in\RR^2_+$ and income $I>0$. The problem is equivalent to learning the ratio $\al_1/\al_2$.
It is obviously optimal for Bob to choose $x=(1/p_1,0)$ if $\al_1/p_1>\al/p_2$ and  $x=(0,1/p_2)$ if $\al_1/p_1<\al/p_2$.

The experimental design presents the subject with a sequence of prices $p$ and incomes $I$, and asks him to choose from $B(p,I)$. Usually only one of the choice problems in the sequence will actually be paid off. It is standard practice in experimental economics to pay out only one of the questions posed to a subject. For the purpose of this example, imagine that the sequence has a length of 2: $(p^1,I^1)$ and $(p^2,I^2)$. Moreover, suppose (again for simplicity) that incomes and prices are such that $I^t = p^t\cdot (1/2,1/2) = 1$, for $t=1,2$. 

Fix the first price at $p^1=(1,1)$. If Alice,  the experimenter, observes a choice of $(1/p_1,0)$ she should conclude that $\al_1/\al_2 > p^1_1/p^1_2$. And given such an inference, it would not make sense to set the second set of prices so that  $p^2_1/p^2_2 < p^1_1/p^1_2$. Alice, following an active learning paradigm of adaptive experimental design, should adjust $p_1/p_2$ upwards. So let us assume that she decides to adjust the ratio $p^1_1/p^1_2$ by a factor of $2$ {\em in the direction in which there is something to learn:} If the choice from $B(p^1,I^1)$ is $(1/p^1_1,0)$, Alice will set $p^2_1/p^2_2 = 2 (p^1_1/p^1_2)$. If the choice is $(0,1/p^1_2)$, she will set $p^2_1/p^2_2 = (1/2)(p^1_1/p^1_2)$. 

Now consider the problem facing our subject, Bob. Suppose that Bob's beliefs are such that  $\al_1<\al_2$, and, to make our calculations simpler, that $\al_1/\al_2\leq 1/2$. If he chooses ``truthfully''  according to his beliefs, he would choose $x=(0,1/p^1_1) = (0,1)$ from $B(p^1,I^1)$ and thus face prices $p^2=(2/3,4/3)$. This means that the relative price of payoffs in state $2$ increase, the state that Bob values the most because he thinks it is the most likely to occur. If instead Bob ``manipulates'' the experiment by choosing $x=(1/p_1,0)$, he will face prices $p^2=(4/3,2/3)$. It is obvious that Bob is better off in the second choice problem from facing the second budget because he will be able to afford a much larger payoff in state $2$.  If Alice only incentivizes (pays out)  the choice from $B(p_2,I_2)$, then Bob is always better off by misrepresenting his choice from the first budget.

If, instead, Alice incentivizes the experiment by implementing one of the choices made by Bob at random (a common practice in economic experiments, see \cite{azrieli2018incentives} for a formal justification), then the utility from truthtelling is  $(1/2)\al_2(1+3/4)=\al_2(7/8)$. The utility from manipulation is $(1/2)(\al_1 + \al_2(3/2))$. As long as $\al_1/\al_2 \in (1/4,1/2)$, the manipulation yields a higher utility than truth telling. 

The convex budgets example illustrates the perils of active learning as a guide to adaptive experimental design, when human subjects understand how the experiment unfolds conditional on how they make choices. The main result of our paper (see Section~\ref{sec:eumodel}) considers belief elicitation, but proposes an active learning algorithm that is based on pairwise comparisons, not choices from convex budgets.

\section{Multiple Price Lists}\label{sec:MPL}

We begin by formally considering the application in the introduction: the use of \emph{Multiple Price Lists} (MPL) to elicit an agent's preferences over risk. MPL was first proposed by \cite{binswanger1981attitudes}, and popularized by \cite{holt2002risk}, who used it to estimate risk attitudes along the lines of the discussion in the sequel. 

We shall consider a version of MPL where a lottery with monetary outcomes is fixed, and an agent chooses between a sure (deterministic) monetary payment $x$ or the lottery. More specifically, consider a lottery where a coin is flipped and if the outcome is heads, the payoff is $\overline{x}$ dollars, while if the outcome is tails, the payoff is $\underline{x}$ dollars. An analyst wants to assess an agent's willingness to participate in the lottery when presented with various deterministic alternatives. Denote this lottery by $L$.

At every round of the experiment, the analyst asks the agent to choose between a deterministic payoff of $x$ or participation in the lottery, and aims to learn the agent's \emph{certainty equivalent}: the deterministic amount that yields indifference. Conventionally (for example, see \cite{holt2002risk,andersen2006elicitation}), the experiment is run by presenting the agent with a list of $n$ pairs $(x_i,L)$. The agent makes a choice from each pair, either the sure amount $x_i$ or the lottery $L$. Then the experimenter draws one of the $n$ questions at random an pays the agent according to the decision he made for that question. (i.e. if he preferred the deterministic amount $x$, he is paid $x$, and otherwise gets to participate in the lottery). We now present a formal model of the MPL experimental design and analyze issues of incentive compatibility.

\subsection{The model}\label{sec:model}
We consider a lottery $L$ with a low outcome $\underline{x}$ and a high outcome $\overline{x}$, $\underline{x} < \overline{x}$. The lottery can operate in any number of ways, for example, by a coin flip. The analyst chooses a discretization $\underline{x} = x_0 < x_1,\ldots < x_{n-1} < x_{n} = \overline{x}$ of the interval $I = [\underline{x}, \overline{x}]$ such that the intervals $I_k = (x_k, x_{k+1}]$ all have equal length $\ell = \frac{\overline{x}-\underline{x}}{n}$. This discretization of $I$ represents the deterministic amounts that the analyst will offer to the agent.

An agent's certainty equivalent is the point $x\in (\underline{x}, \overline{x})$ such that he is indifferent between receiving $x$ versus participating in the lottery.  Certainty equivalents will be uniquely determined by an agent's utility over money, as long as his utility function is strictly increasing. 

For example, if an agent values money according to $u:\RR\to\RR$, his certainty equivalent (assuming that $L$ is a coin-flip) would be the point $x\in (\underline{x}, \overline{x})$ such that $u(x) = \frac{1}{2}u(\underline{x}) + \frac{1}{2}u(\overline{x})$. In our model, we consider agents whose utility functions belong to a given family $\mathcal{U}$ of functions such that a given certainty equivalent uniquely determines the utility function of the agent, and vice-versa. For example, if $\mathcal{U} = \{x\mapsto x^{\sigma} : 0 < \sigma < 1\}$, so utilities take the CRRA form we discussed in the introduction, then  $\sigma$ uniquely determines the point $x$ such that $x^{\sigma} = \frac{1}{2}\underline{x}^{\sigma} + \frac{1}{2}\overline{x}^{\sigma}$.

\subsection{Sequential Search}

We first consider a simple mechanism that aims to find the agent's certainty equivalent by performing a sequential search on $x_1,\ldots, x_{n}$. On round $t$ of the experiment, the agent chooses between the lottery and a deterministic payoff of $x_t$. If he chooses the lottery, the experiment continues, and if he chooses $x_t$ or claims indifference, the experiment terminates. If the experiment terminates at round $T$, the analyst can conclude that the agent reported a certainty equivalent lying in the interval $(x_{T-1}, x_T]$.

The goal of the analyst is to make a payment to the agent at the end of the experiment such that the agent is incentivized to answer questions according to his true certainty equivalent. We analyze a common scheme used in experiments: if the experiment terminates after $T$ rounds, choose $t$ randomly from $\{1,\ldots, T\}$, and pay the agent based on his preference on the $t$th question: so if $t = T$, the agent receives $x_T$, otherwise the agent receives a payment that is the outcome of the lottery. However, as discussed in the introduction, this scheme is not incentive compatible. Indeed, if $x_T$ is the agent's true certainty equivalent, he has a profitable deviation to push the experiment to terminate at $x_{T+1}$. The agent is indifferent between receiving $x_T$ and participating in the lottery, so by declaring a certainty equivalent that is higher than $x_T$ he may possibly win an amount larger than $x_T$, and which he values strictly more than the lottery. 

We now show that under some simplifying assumptions, this kind of payment scheme can at least be implemented in a manner such that the agent's true certainty equivalent can be accurately inferred based on his report. Let $\mathcal{U}$ be a family of utility functions such that each $u\in \mathcal{U}$ satisfies an inverse Lipschitz condition with constant $K_u$: for all $x, x^{*}$, $|u(x) - u(x^*)| > K_u |x-x^*|$. Let $K = \sup_{u\in \mathcal{U}} K_u$. Finally, let $M = \sup_{u\in \mathcal{U}} (u(\overline{x}) - u(\underline{x}))$.

For $t\in\{1,\ldots, n\}$ let $p_t$ denote the probability that the agent is paid the deterministic amount $x_{t}$ if the experiment stops on round $t$ (so the agent participates in the lottery with probability $1-p_t$). An agent with true certainty equivalent $x$ and corresponding utility function $u$ has an expected payoff of $$\Payoff(x, x_t) = p_t u(x_t) + (1-p_t) u(x) $$ for reporting a certainty equivalent in $(x_{t-1}, x_t]$.

Let $r:(\underline{x}, \overline{x})\to \{x_1,\ldots, x_n\}$ be the best response of an agent with certainty equivalent $x$: $$r(x) = \argmax_{x_t} \Payoff(x, x_t).$$ We refer to $r$ as the report function.

We now show that when $p_1,\ldots, p_{n}$ satisfy $p_{t+1} < p_t$\footnote{This is true for the standard uniform randomization scheme, as $p_t = 1/t$. More generally, without this condition the agent will have incentives to report high certainty equivalents as this is not penalized by lower probabilities of winning the certain amount.} and $p_{t+1} < \frac{K\ell}{2M}p_t$, the analyst can recover the agent's true certainty equivalent up to some low error via the sequential search mechanism. 

We should emphasize that the agent will not be truthful, in the sense of reporting their true certainty equivalent. However, we are still able to back out the true certainty equivalent from understanding the agents strategic incentives. 

We proceed in steps. First, we characterize the best responses for agents with certainty equivalents belonging to $\{x_1,\ldots, x_{n-1}\}$. We show that the report function is one to one. This implies that $r(x_t) = x_{t+1}$, since $r(x_t) > x_t$, for each $t = 1,\ldots, n-1$.

\begin{proposition}
For $x_t\in\{x_1,\ldots, x_{n-1}\}$, $r(x_t) = x_{t+1}$.
\end{proposition}

\begin{proof}
Note that since $r(x_{n-1}) = x_n$, it suffices to show that $r$ is injective on $\{x_1,\ldots, x_{n-1}\}$.

Suppose $r(x_{t_1}) = r(x_{t_2})$, and without loss of generality let $t_1 \ge t_2$. Let $r(x_{t_1}) = r(x_{t_2}) = x_t$. Since any agent is incentivized to report higher than their true certainty equivalent, $t > t_1, t_2$, so in particular $t\ge t_1+1$. Let $u_1$ denote the utility function of the agent with true type $x_{t_1}$, $u_2$ that of the agent with true type $x_{t_2}$.

For any $s\neq t$ we have (since $r$ gives the best response):

$$\begin{aligned}
    p_tu_1(x_t) - p_t u_1(x_{t_1}) & > p_s u_1(x_s) - p_s u_1(x_{t_1}), \\
    p_t u_2(x_{t}) - p_t u_2(x_{t_2}) & > p_s u_2(x_s) - p_s u_2(x_{t_2}).
\end{aligned}$$

Adding the two inequalities and rearranging gives
\begin{equation}\frac{p_t}{p_s} > \frac{u_1(x_{s}) - u_1(x_{t_1}) + u_2(x_{s}) - u_2(x_{t_2})}{u_1(x_{t}) - u_1(x_{t_1}) + u_2(x_{t}) - u_2(x_{t_2})}.\label{eq1}\end{equation}

At $s = t_1$ (we know $t_1\neq t$, since $r(x_{t_1}) > x_{t_1}$), Equation \ref{eq1} simplifies to $$\frac{p_t}{p_{t_1}} > \frac{u_2(x_{t_1}) - u_2(x_{t_{2}})}{u_1(x_{t}) - u_1(x_{t_1}) + u_2(x_{t}) - u_2(x_{t_2})} > \frac{u_2(x_{t_1}) - u_2(x_{t_2})}{2M},$$ so $$u_2(x_{t_1}) - u_2(x_{t_2}) < 2M\frac{p_t}{p_{t_1}}\le 2M\frac{p_{t_1+1}}{p_{t_1}} < K\ell.$$ The inverse Lipschitz condition on $u_2$ then implies that $|x_{t_1} - x_{t_2}| < \ell$, which cannot happen unless $t_1 = t_2$.
\end{proof}

Thus, if an agent's true certainty equivalent happens to coincide with one of the points of the discretization, the agent will answer questions as if his certainty equivalent is the next point in the discretization.

For the next step, we need an additional Lipschitz type condition on utility functions. Suppose there are constants $C_1$ and $C_2$ such that for any $x, x^*\in(\underline{x}, \overline{x})$, with $u, u^*$ the corresponding utility functions, and for any $x, x''\in (\underline{x}, \overline{x})$, $$|u(x') - u^*(x'')| \le C_1|x-x^*| + C_2|x' - x''|.$$ Moreover, let $$\lambda = \inf_{x\in (\underline{x}, \overline{x})} \min_{s, t} \left|\Payoff(x, x_s) - \Payoff(x, x_t)\right|,$$ be the smallest possible deviation in payoff obtained by changing one's report.

We also require the assumption that if $x^*\ge x$, then $u^*(x')\ge u(x')$ for any $x'$, where $u^*$ and $u$ are the utility functions corresponding to certainty equivalents $x^*$ and $x$, respectively. This is an intuitive condition stating that agents with a higher certainty equivalent value money more than agents with a lower certainty equivalent (note that the CRRA utilities discussed previously satisfy this property). This in particular implies that $u^*(x^*)\ge u^*(x)\ge u(x)$, so $\Payoff(x^*, x_k)\ge\Payoff(x, x_k)$ for any $x_k$ in the discretization. We will use this in the proof of the following proposition, which establishes that $r$ satisfies a certain weak monotonicity property.

\begin{proposition} Let $x_{t-1} < x < x^* \le x_t$ with $x^* - x < \frac{\lambda}{2(2C_1+C_2)}$, and suppose $r(x^*) \le x_{t+1}$. Then, $r(x)\le x_{t+1}$.
\end{proposition}

\begin{proof}

Let $u, u^*$ be the utility functions corresponding to certainty equivalents $x$ and $x^*$, respectively. We first bound the increase in payoff an agent of type $x^*$ experiences over an agent of type $x$ for making the same report. For any $x_k$, we have
$$\begin{aligned}
\Payoff(x^*, x_k)-\Payoff(x, x_k) &= p_k (u^*(x_k) - u(x_k)) - p_k  (u^*(x^*) - u(x)) + (u^*(x^*) - u(x)) \\
& < p_k(u^*(x_k) - u(x_k)) + (u^*(x^*) - u(x)) \\
&< (u^*(x_k) - u(x_k)) + (u^*(x^*) - u(x)) \\
&\le C_1\ell + (C_1+C_2)\ell \\
&\le \frac{\lambda}{2}
\end{aligned}$$

As $r(x^*)\le x_{t+1}$, either $r(x^*) = x_{t+1}$ or $r(x^*) = x_{t}$. Suppose $r(x^*) = x_{t+1}$. We show that an agent of type $x$ cannot increase his payoff by reporting above $x_{t+1}$. Let $s > t+1$.

Plugging $x_{t+1}$ into the above bound gives $$\Payoff(x^*, x_{t+1})\le\frac{\lambda}{2} + \Payoff(x, x_{t+1}),$$ and the definition of $\lambda$ gives that $$\Payoff(x^*, x_{t+1}) \ge \Payoff(x^*, x_s) + \lambda.$$ Combining the two inequalities yields $$\begin{aligned}\Payoff(x, x_{t+1})\ge \Payoff(x^*, x_{t+1}) - \frac{\lambda}{2}  &>\Payoff(x^*, x_s) + \frac{\lambda}{2} \\
&> \Payoff(x^*, x_s) \\
&> \Payoff(x, x_s),\end{aligned}$$ so $r(x)\le x_{t+1}$.


In the case that $r(x^*) = x_t$, we similarly get $r(x)\le x_t$.
\end{proof}

We can then repeatedly apply this proposition starting with $r(x_t) = x_{t+1}$ to conclude that for any $x_{t-1} < x \le x_t$, we have $r(x)\le x_{t+1}$.

Putting things together, we get:

\begin{theorem}\label{thm:back-out} If $r(x) = x_{t+1}$, then $x_{t-1} < x < x_{t+1}$. \end{theorem}

Thus, to learn the agent's true certainty equivalent to within $\varepsilon$-error, the analyst chooses a discretization with $\frac{\overline{x}-\underline{x}}{n}\le\frac{\varepsilon}{2}$, and runs a sequential search over the discretization. The number of questions the analyst asks is $O(\frac{1}{\varepsilon})$. 

Of course, to lower the number of questions asked, the analyst could instead perform a binary search. It is easy to see that, like in the sequential search mechanism, simply implementing a uniformly random question is not incentive compatible. For example, consider a discretization with deterministic amounts $x_1,\ldots, x_{7}$, and consider an agent with true certainty equivalent at $x_3$. For simplicity, we assume that if when presented with $(x_i, L)$ the agent is indifferent between $x_i$ and $L$, he chooses $x_i$. If the agent answers truthfully, the pairs offered by a binary search would be $(x_4, L)$, $(x_2, L)$, and $(x_3, L)$, and his choices would have been $x_4$, $L$, and $x_3$, respectively. The agent's expected payoff is $(1/3)u(x_4) + (1/3)u(L) + (1/3)u(x_3) = (1/3)u(x_4) + (2/3)u(L)$. Suppose instead the agent answers as if his true certainty equivalent is $x_5$. Then, the pairs he gets offered would be $(x_4, L)$, $(x_6, L)$, and $(x_5, L)$, and his choices would have been $L$, $x_6$, and $x_5$, respectively. His expected payoff is then $(1/3)u(L) + (1/3)u(x_6) + (1/3)u(x_5)$, which is clearly a profitable deviation.

It is unclear if this scheme can be directly modified to satisfy incentive compatibility properties, but since the payments in the sequential search mechanism only depended on the last question asked, we can use the same payment scheme here so that Theorem \ref{thm:back-out} holds. So now the analyst can learn the agent's certainty equivalent to within an error of $\varepsilon$ with $O(\log 1/\varepsilon)$ questions.

\section{General Preference Elicitation}\label{sec:general}

Our discussion so far has focused on a specific, albeit ubiquitous, preference elicitation environment. In the rest of the paper we introduce a general model of incentive compatible active learning. We introduce the idea of incentive compatible query complexity: the sample size that guarantees some learning objective while maintaining incentive compatibility. 

The main application of our tools will be to expected utility theory. We shall introduce a learning algorithm that is incentive compatible for learning the beliefs of an agent that has expected utility preferences. 

We focus on learning an agent's preferences. The agent will be modeled as having a utility function parameterized by some type, which generates the agent's choices, that the learner wishes to infer. To this end,  $\Theta$ is a type space equipped with a metric $d:\Theta\times\Theta\to\RR_{\ge 0}$ that is bounded with respect to $d$. $\mathcal{O}$ is the space of possible outcomes. An agent of type $\theta\in\Theta$ has utility  $u(\theta, o)$ if the outcome is $o\in\mathcal{O}$. $\theta$ induces a preference relation $\succsim$ over $\mathcal{O}$ defined by $o\succsim o'\iff u(\theta, o)\ge u(\theta, o')$.

An analyst aims to learn the agent's type by asking him to make a sequence of choices between pairs of outcomes.\footnote{One can imagine many other protocols for learning. We constrain ourselves to protocols that are based on a sequence of pairwise comparisons. Such protocols are common in practice, and are the obvious empirical counterpart to the decision theory literature in economics and statistics. This stands in contrast with the literature on scoring rules, which allows for richer message spaces.} The agent makes choices among the pairs presented to him. 

The agent's choices can be thought of as the result of a strategy. Formally, a strategy $\sigma$ is a mapping $$\sigma:\bigcup_{t} \{((o_1, o_1'), \mathbf{1}_{o_1\succsim o_1'}),\ldots, ((o_t, o_t'), \mathbf{1}_{o_t\succsim o_t'}), (o_{t+1}, o_{t+1'})\}\to\Delta\{0, 1\}$$ that dictates a (potentially randomized) response for every possible history of the interaction up to any given time. Let $\Sigma$ denote the collection of all possible consistent strategies (a strategy is consistent if its outputs up to any given time are consistent with some preference relation in the type space).

For any strategy $\sigma$, let $\hat{\sigma}$ denote an oracle with memory that responds to queries of the form ``is $o$ preferred to $o'$?'' according to $\sigma$ given the history of previous queries made so far. Let $\hat{\Sigma} = \{\hat{\sigma} : \sigma\in\Sigma\}$ denote the collection of oracles corresponding to all possible strategies. For a type $\theta\in\Theta$, let $\hat{\theta}\in\hat{\Sigma}$ denote the oracle that responds truthfully according to $\theta$ (i.e. on query $(o, o')$ it returns $\mathbf{1}_{u(\theta, o)\ge u(\theta, o')}$).

We imagine the oracle playing the role of the agent: in an interaction with the analyst, an agent of true type $\theta$ chooses to act as an oracle for some strategy $\sigma$ (departing from standard terminology, we allow the oracle to have randomized responses).

The analyst implements a learning mechanism, which consists of the following steps: 

\begin{enumerate}
    \item Run a (potentially randomized) learning algorithm $\mathcal{A}:\hat{\Sigma}\to\Theta$ that has access to oracle $\hat{\sigma}$ and can make queries to $\hat{\sigma}$ of the form $(o, o')$ for $o, o'\in\mathcal{O}$.
    \item Arrive at a hypothesis $\theta^h\sim\mathcal{A}(\hat{\sigma})$ for the agent's type.
    \item Implement the agent's response on the last query.\footnote{Within adaptive experimental design, the idea of making a last choice on behalf of the agent is due to Ian Krajbich.}
\end{enumerate}

We now establish the notion of learnability that we work with. This definition is not concerned with issues of incentive compatibility: it is simply a refinement of the standard notion of a learning algorithm that stipulates that we learn a truthfully reported hypothesis accurately. Since in our setting the analyst has full control over the data he learns from, our requirements on the error of the algorithm are with respect to the metric $d$ on the space of types $\Theta$. 

\begin{definition}
$\mathcal{A}:\hat{\Sigma}\to\Theta$ is an $(\varepsilon, \delta)$-learning algorithm if for all $\theta\in\Theta$, $$\Pr_{\theta^h\sim\mathcal{A}(\hat{\theta})}[d(\theta, \theta^h)\le\varepsilon]\ge 1-\delta.$$ The number of queries made by $\mathcal{A}$ to the oracle $\hat{\theta}$ is the \emph{query complexity} of $\mathcal{A}$, denoted by $q(\varepsilon, \delta)$.
\end{definition}

Next, we define what it means for a learning algorithm to be incentive compatible. Intuitively, we require that if the learning algorithm is terminated on round $T$, and the analyst implements the agent's preferred outcome on the $T$th query $(o_T, o_T')$, then the agent (with high probability) cannot gain a non-negligible advantage over truthfully reporting by attempting to answer questions strategically. Let $\mathcal{A}_T(\hat{\sigma}) = (o_T, o_T')\in\Delta(\mathcal{O}\times\mathcal{O})$ denote the $T$th query to $\hat{\theta}$ made by an execution of $\mathcal{A}(\hat{\sigma})$.

\begin{definition}
$\mathcal{A}:\hat{\Sigma}\to\Theta$ is $(\tau, \nu)$-incentive compatible if there exists a $T(\tau, \nu)\in\NN$ such that for all $T\ge T(\tau, \nu)$, the following holds for any type $\theta$ and strategy $\sigma$: $$\Pr_{\substack{(o_T, p_T)\sim \mathcal{A}_T(\hat{\theta}) \\ (o_T', p_T')\sim \mathcal{A}_T(\hat{\sigma})}}[u(\theta, q_T)\ge u(\theta, q_T') - \tau]\ge 1-\nu,$$ where $q_T$ ($q_T'$) is the preferred outcome between $o_T$ and $p_T$ ($o_T'$ and $p_T'$) according to oracle $\hat{\theta}$ ($\hat{\sigma}$). The quantity $T(\tau, \nu)$ is the \emph{IC complexity} of $\mathcal{A}$.\footnote{In this definition, $(o_T, p_T)$ and $(o_T', p_T')$ are drawn from independent executions of $\mathcal{A}$.}
\end{definition}

Our goal is to design mechanisms that learn the agent's true type in an incentive compatible manner.

\begin{definition}$\mathcal{A}:\hat{\Sigma}\to\Theta$ is an $(\varepsilon, \delta, \tau, \nu)$-IC learning algorithm if it is an $(\varepsilon, \delta)$-learning algorithm that is $(\tau, \nu)$-IC. We refer to the quantity $\max(q(\varepsilon, \delta), T(\tau, \nu))$ as the \emph{IC learning complexity} of $\mathcal{A}$.\end{definition}

\subsection{An incentive compatible exhaustive search}

We first give a very simple method of achieving incentive compatible learning in the general framework introduced in Section~\ref{sec:general}. The method proceeds by exhaustively searching over the type space, and requires a simple structural assumption.  The assumption connects agents' payoffs to the distance metric used by the learner to assess learning accuracy. In a sense, this lines up the agent's incentives with the learner's objective, and makes it easy to obtain a satisfactory algorithm.

Suppose there exists a one-to-one assignment of outcomes to types $s:\Theta\to\mathcal{O}$ such that $$u(\theta, s(\theta')) > u(\theta, s(\theta''))\iff d(\theta, \theta') < d(\theta, \theta''),$$ so in particular $\theta = \argmax_{\theta'} u(\theta, s(\theta'))$. In the literature on scoring rules, $s$ is called \emph{effective} with respect to $d$ \cite{friedman1983effective}.

The following is an incentive compatible learning algorithm. Recall that an $\varepsilon$-cover of a subset $K$ of a metric space $(M, d)$ is a set of points $C$ such that for every $x\in K$, there is an $x^{*}\in C$ such that $d(x, x^*)\le\varepsilon$.

\begin{enumerate}
    \item Initialize an $\varepsilon$-cover $\{\theta_1,\ldots\}$ of $\Theta$ with respect to $d$.
    \item Initialize $\theta^h\leftarrow\theta_1$.
    \item For $t = 1$ to $T$:
    \begin{enumerate}
        \item Query $(s(\theta_{t+1}), s(\theta^h))$.
        \item If $s(\theta_{t+1})$ is preferred, $\theta^h\leftarrow\theta_{t+1}$.
    \end{enumerate}
    \item Output $\theta^h$.
    \item Pay the agent $s(\theta^h)$.
\end{enumerate}

By definition of the function $s$, allowing the algorithm to exhaustively search over all points of the cover will yield a $\theta^h$ that is the most preferred point in the cover and is also the closest point in the cover to the report. So reporting $\theta$ yields $d(\theta, \theta^h)\le\varepsilon$. Moreover, this (deterministic) algorithm satisfies $(0, 0)$-incentive compatibility for any runtime $T$, since lying at any round would simply reduce the payoff of stopping at any round. The learning complexity is the covering number $N_{\varepsilon}(\Theta)$ of the type space (which is finite as $\Theta$ is bounded).


We now present some natural preference environments in which such an assignment function can be constructed. In the following discussion, the outcome space is $\mathcal{O} = \RR^n$, and $\Theta$ is assumed to be bounded so that the search above terminates.

\begin{itemize}
    \item \emph{Euclidean preferences}. Each agent has an ``ideal point'' $\theta\in\RR^n$, and $u(\theta, x)\ge u(\theta, y)$ iff $||x-\theta||\le ||y-\theta||.$ Let $s:\RR^n\to\RR^n$ be the identity.
    \item \emph{Linear preferences}. The type of an agent is a vector $\theta\in\RR^n$ and $u(\theta, x)\ge u(\theta, y)$ iff $\theta.x\ge\theta.y$ (in order for preferences to be distinguishable we assume that no two $\theta, \theta'\in\Theta$ are scalar multiples of one another, and so for simplicity we normalize so that all types have the same length). The indifference sets of an agent of type $\theta$ are the hyperplanes $\{x : \theta.x = k\}$, for $k\in\RR$. For each $\theta$, there is a unique indifference set that is tangent to the unit $(n-1)$-sphere $S^{n-1}$. Let $s(\theta)$ be that tangent point.
\end{itemize}

Euclidean and linear preferences are characterized by natural axioms for preference relations \cite{chambers2019spherical}.

More generally, suppose the preferences of each agent are continuous and strictly convex, which we define as the upper contour sets $C(x) = \{y:y\succsim x\}$ being closed, convex, and any supporting hyperplane of $C(x)$ being unique, for all $x$. For a type $\theta$, real number $k$, and outcome $x\in\RR^n$ such that $u(\theta, x) = k$, let $H_k^{\theta}(x)$ denote the supporting hyperplane of the upper contour set $\{y: u(\theta, y)\ge k\}$ at $x$.

Suppose that the following uniqueness requirement holds: for every pair of types $\theta\neq\theta'$, real number $k$, and outcome $x\in\RR^n$ such that $u(\theta, x) = k$, if $k'$ is such that $u(\theta', x) = k'$, it holds that $H_k^{\theta}(x) \neq H_{k'}^{\theta'}(x)$. We call this property \emph{hyperplane uniqueness}.\footnote{Hyperplane uniqueness is reminiscent of the single-crossing property in mechanism design.} Then, the argument for IC learnability in the case of linear preferences can be adapted to this setting as well. Though the assignment function $s$ we construct may not necessarily be effective, we show that exhaustively searching over a sufficiently fine cover is nevertheless incentive compatible.

\begin{thmn}[\ref{thm:strictly_convex}]
Let $\Theta$ be a type space such that the preferences induced by each $\theta\in\Theta$ are continuous, strictly convex, and satisfy hyperplane uniqueness. Then, there exists a metric on $\Theta$ with respect to which $\Theta$ is $(\varepsilon, 0, 0, 0)$-IC learnable.
\end{thmn}

\begin{proof}
For each $\theta$, let $s(\theta)$ be the unique maximizer of $u(\theta,x)$ over the unit $(n-1)$-sphere $S^{n-1}$; uniqueness follows from the strict convexity of preferences. Note that $S^{n-1}$ and  $C = \{y : u(\theta, y)\ge k\}$, with $k=u(\theta,s(\theta))$, are on differing sides of the supporting hyperplane of $C$ at $s(\theta)$.  Hyperplane uniqueness ensures that $s$ is one-to-one. Let $d$ be the metric where $d(\theta, \theta')$ is the Euclidean distance between $s(\theta)$ and $s(\theta')$.

Let $C^{\theta}_{k'} = \{y : u(\theta, y) > k'\}$ with $k'<k$ and let $N^{\theta}_{k'} = C_{k'}\cap S^{n-1}$. Clearly the diameter of $N^{\theta}_{k'}$ converges to $0$ as $k'\uparrow k$. Therefore, for each $\theta$ and $\varepsilon > 0$, there is an open neighborhood $N^{\theta}_{\varepsilon}\subset B_{\varepsilon}(s(\theta))$ of $s(\theta)$ such that $u(\theta, x)> u(\theta, y)$ for all $x\in N^{\theta}_{\varepsilon}$ and all $y\in S^{n-1}\setminus N^{\theta}_{\varepsilon}$ (where $N^{\theta}_{\varepsilon}$ is of the form $N^{\theta}_{k'}$, for $k'$ sufficiently close to $k$). \footnote{The neighborhoods and balls are with respect to the subspace topology on $S^{n-1}$.}

Now, fix the learning parameter $\varepsilon$, and let $\eta > 0$ be sufficiently small such that if $K$ is an $\eta$-cover of $S^{n-1}$, $K\cap N^{\theta}_{\varepsilon}\neq\emptyset$ for all $\theta$. Then, any $x\in K\cap N_{\varepsilon}^{\theta}$ satisfies $u(\theta, x) > u(\theta, y)$, for all $y\in S^{n-1}\setminus B_{\varepsilon}(s(\theta))$. Thus, the most preferred point of an agent of type $\theta$ is contained in $K\cap N_{\varepsilon}^{\theta}\subset B_{\varepsilon}(s(\theta))$, and so exhaustively searching over this $\eta$-cover is an $\varepsilon$-learning algorithm with respect to $d$ that is incentive compatible. \end{proof}

It is an interesting question to see what structural conditions one can impose on the type space, the outcome space, etc. to write down better learning mechanisms. For example, one might hope to achieve a learning complexity that is logarithmic in the size of the cover $N_{\varepsilon}(\Theta)$. As we will see in the case of expected utility, naive learning algorithms achieve this sample complexity, but fail to be incentive compatible. More generally one can ask if there is a combinatorial complexity measure (such as VC dimension in the case of PAC learning) that characterizes the complexity of incentive compatible learning. 

\subsection{The expected utility model of choice under uncertainty}\label{sec:eumodel}

We now turn to the case of belief elicitation for an expected utility agent. Belief elicitation has a long history in experimental economics, and in the theoretical literature on scoring rules (e.g.\  \cite{chambers2017dynamic}; see \cite{conitzer2009prediction} for a survey). A major difference with the theory of scoring rules is that we shall take as given a protocol that is based on pairwise comparisons among uncertain prospects.\footnote{This follows experimental practice, as well as the standard model of choice under uncertainty; starting from von-Neumann and Morgenstern \cite{vonneumann2007theory} and Savage \cite{savage1972foundations}. In the scoring rule model, subjects are asked to report beliefs rather than carrying out a sequence of binary choices. In any case we shall use scoring rules in our solution, just not by asking subjects to report their beliefs.} The case of passive learning was studied in \cite{basu2018learnability}.

There are $n$ states of the world, indexed by $i = 1,\ldots, n$. An agent has a subjective belief $\alpha\in \Delta_n$, where $\alpha_i$ is the probability the agent assigns to state $i$ occurring. The agent evaluates the payoff of a vector of rewards $x\in\RR^n$ by computing expectation according to $\alpha$. An agent's belief $\alpha$ defines a preference relation $\succsim\subseteq\RR^n\times\RR^n$, where $$x\succsim y\iff \alpha.x\ge\alpha.y.$$

An analyst would like to learn $\alpha$ by asking the agent to make several choices between vectors of rewards.  The analyst presents the agent with a sequence of pairs $(x,y)$ and if the agent chooses $x$ she infers that $(x-y).\alpha\geq 0$. So the problem is related to that of learning half spaces, but with the added complication of having to respect incentive compatibility. An important assumption is that the analyst is able to simulate the states of the world and observe a state according to the ``ground truth'' process governing the states (so for example if the states were ``rain'', ``snow'', and ``shine'', the analyst could simply observe the weather on the given day).

Using the notation of the previous section, $\Theta = \Delta_n$, $\mathcal{O} = \RR^n$, and $u(\alpha, x) = \E_{i\sim\alpha}[x] = \alpha.x$.

In the context of learning the agent's true belief, the analyst uses total variation distance $||\alpha-\beta||_{TV} = \frac{1}{2}\textstyle\sum_{i=1}^n|\alpha_i - \beta_i|$ to measure accuracy/error.

\subsubsection{Naive algorithms are not incentive compatible}

First, to illustrate the restrictions of our definitions, we write down a naive algorithm for eliciting $\alpha$ that achieves a good query complexity, but is not incentive compatible.

Consider a mechanism that tries to elicit each $\alpha_i$ by performing a search (sequential or binary) on each state. That is, for each state $i$, the algorithm makes queries $(e_i, c_i\vec{1})$, varying $c_i$ over a $\tfrac{2\varepsilon}{n}$-cover of $[0, 1]$ to find the indifference points, which reveal $\alpha_i$ to within an error of $\tfrac{2\varepsilon}{n}$. So, for example, a binary search uses $O(n\log\frac{n}{\varepsilon})$ questions to arrive at a hypothesis within total variation distance $\varepsilon$ from $\alpha$. Note that a $\tfrac{2\varepsilon}{n}$-cover of the simplex $\Delta_n$ with respect to total variation distance contains $O((n/\varepsilon)^n)$ elements, so performing a state-wise binary search exponentially improves upon a search over the entire cover.

However, incentive compatibility is broken rather easily, since the agent has a great deal of control over what questions the agent asks (in a similar manner to the situation in MPL). Consider the following simple example: suppose the analyst fixes a discretization of $[0, 1]$ with sure amounts $x_1,\ldots, x_7$, as in the binary search MPL example from Section~\ref{sec:MPL}, and suppose an agent has a true belief $\alpha$, with $\alpha_n\le x_6$. If, instead of $\alpha$, the agent reports an $\alpha'$ with $\alpha'_n\in (x_6, x_7)$, the final question he would get asked would be $(e_n, x_7\vec{1})$. The agent would prefer $x_7\vec{1}$, and thus would get paid off a sure amount of $x_7$. It is clear that truthfully reporting yields a strictly lower payoff than the misrepresentation. Notice that this situation is even worse than that of MPL, since if the binary search ends on state $n$, then regardless of the probabilities an agent assigns to states $1,\ldots, n-1$, he will want to answer questions as if he assigns most weight to state $n$ -- so there is no hope of backing-out an agent's true belief using this kind of scheme.

A strategic agent can easily outwit minor modifications to this scheme: for example if the analyst does the binary searches in a random order over the states, the agent can adaptively report a belief that assigns most weight to the last state over which the analyst performs a binary search. 

\subsubsection{A mechanism based on scoring rules.}

In this section we present an IC learning algorithm with IC learning complexity $$O\left(n^{3/2}\log n\max\left(\log\frac{n}{\varepsilon}, \log\frac{1}{\tau}\right)\right).$$ The algorithm is based on ideas from active learning, and specifically leverages convergence bounds on so-called disagreement based methods. Let $||\cdot||$ denote the $L^2$ norm, let $S^{n-1}$ denote the unit $(n-1)$-sphere, and let $\rho:\RR^n\to S^{n-1}$ denote the projection map onto the unit sphere defined by $\rho(\alpha) = \frac{\alpha}{||\alpha||}$. 

We now present an incentive compatible learning algorithm that we henceforth refer to as $\mathcal{A}$.

\begin{enumerate}
    \item Initialize $\mathcal{H}^0\leftarrow \Delta_n$.
    \item For $t = 1$ to $T$:
    \begin{enumerate}
        \item Choose $v$ uniformly at random from $S^{n-1}$. If the hyperplane $\{x:v.x = 0\}$ does not intersect $\mathcal{H}^{t-1}$, resample.
        \item Let $\beta^1, \beta^2$ be any elements of $\mathcal{H}^{t-1}$ such that $\rho(\beta^1) - \rho(\beta^2)$ is a scalar multiple of $v$.
        \item Query oracle on pair $(x_t = \rho(\beta^1), y_t = \rho(\beta^2))$.
        \item $\mathcal{H}^{t}\leftarrow\mathcal{H}^{t-1}\cap \{\beta\in\Delta_n : \beta \text{ is consistent with label on } (x_t, y_t)\}$
    \end{enumerate}
    \item Output any $\beta^h\in \mathcal{H}^T$.
    \item Pay the agent off based on preference from $(x_T, y_T)$. If $z_T$ is the preferred vector, simulate states of the world, and pay $(z_T)_i$ if state $i$ occurs.
\end{enumerate}

Before analyzing the algorithm, let us briefly remark that the analyst can always find $\beta^1$, $\beta^2$ satisfying the required conditions to query the agent. Let normal vector $v\in S^{n-1}$ define a hyperplane $v.x = 0$ that cuts through the projection $\rho(\mathcal{H})\subset S^{n-1}$ of the current hypothesis set onto the unit sphere. Let $w$ be a point in the interior of $\rho(\mathcal{H})$ such that $v.w = 0$.\footnote{The interior of $\rho(\mathcal{H})$ can be written as $\rho(\{\beta\in\Delta_n : v_1.\beta > 0, \ldots, v_T.\beta > 0\})$ for some $v_1,\ldots, v_T$, which is a non-empty intersection of open half-spaces as the agent's responses are required to be consistent.} We can find an open ball $B(w, r)$ (with respect to the subspace topology on $S^{n-1}$ induced by $\RR^n$) of radius $r$ centered at $w$ such that $B(w, r)\subset \rho(\mathcal{H})$. Then, take a point $x\in B(w, r)$ in the positive $v$ direction from $w$ and $y\in B(w, r)$ in the negative $v$ direction from $w$ such that $||x-w|| = ||y-w||$. Then, $x - y = v||x-y||$. 

Choosing $\beta^1$ and $\beta^2$ in this manner has no effect on the analysis of the learning rate, but is the main ingredient in achieving incentive compatibility. The learning guarantees we obtain are due to standard bounds on the label complexity of disagreement based active learning.

\begin{theorem}\label{thm:eu_learning}
$\mathcal{A}$ is a learning algorithm of query complexity $O(n^{3/2}\log n\log\frac{n}{\varepsilon})$ with respect to total variation distance. 
\end{theorem}

\begin{proof}Suppose $\mathcal{A}$ receives as input an oracle $\hat{\alpha}$. If on a given round we sample a normal vector $v$ and correspondingly query points $(x_t, y_t)$, the truthful agent's/oracle's preference from $(x_t, y_t)$ precisely reveals $\sgn(v.\alpha)$ -- this is simply because $x_t-y_t$ and $v$ determine the same hyperplane.

The VC dimension of the expected utility model is linear (Theorem 2 of \cite{basu2018learnability}), and the disagreement coefficient of the class of homogeneous linear separators with respect to the uniform distribution over normal vectors is bounded above by $\pi\sqrt{n}$ (Theorem 1 of \cite{hanneke2007bound}). Standard convergence results in active learning (see, e.g., \cite{dasgupta2011two}) then imply that with $O(n^{3/2}\log n\log\frac{1}{\eta})$ queries, it holds with high probability that $\err_{\alpha}(\alpha^h)\le \eta$ for all $\alpha^h$ in the final hypothesis set, where $$\err_{\alpha}(\beta) = \Pr_{v\sim S^{n-1}}[\sgn(v.\alpha)\neq\sgn(v.\beta)] = \frac{\arccos(\rho(\alpha).\rho(\beta))}{\pi}.$$

For $\varepsilon > 0$, let $\eta = \frac{2\varepsilon}{\pi n} < \frac{1}{\pi}\arccos\left(1-\frac{2\varepsilon^2}{n^2}\right)$, so $\cos(\pi\eta) > 1-\frac{2\varepsilon^2}{n^2}$.

Running $\mathcal{A}$ for $O(n^{3/2}\log n\log\frac{n}{\varepsilon})$ rounds yields that for any hypothesis $\alpha^h\in \mathcal{H}^T$, $$\begin{aligned}||\rho(\alpha) - \rho(\alpha^h)|| &= \sqrt{(\rho(\alpha) - \rho(\alpha^h)).(\rho(\alpha) - \rho(\alpha^h))} \\ &= \sqrt{2-2\rho(\alpha).\rho(\alpha^h)}\\ &\le\sqrt{2-2\cos(\pi\eta)},\end{aligned}$$ which is at most $2\varepsilon/n$ (where the final inequality is with high probability over the execution of $\mathcal{A}$). Thus $||\alpha-\alpha^h||\le\frac{2\varepsilon}{n}$, and so $||\alpha-\alpha^h||_{TV}\le\varepsilon$. \footnote{$||\alpha-\alpha^h|| \le \frac{2\varepsilon}{n}\implies \sum_{i=1}^n (\alpha_i - \alpha_i^h)^2\le \frac{4\varepsilon^2}{n^2}$, so $(\alpha_i-\alpha_i^h)\le\frac{2\varepsilon}{n}$ for each $i$, which implies that $||\alpha-\alpha^h||_{TV}\le\varepsilon$.}
\end{proof}

We now analyze incentive compatibility properties of the algorithm. The main ingredient is in using the mapping $(\alpha, i)\mapsto \rho(\alpha)_i = \frac{\alpha_i}{||\alpha||}$ to choose what questions to ask. This mapping is known as the \emph{spherical scoring rule}, and incentivizes truthful forecasts, in the sense that $\alpha = \argmax_{\beta}\E_{i\sim\alpha}[\rho(\beta)_i]$. The spherical scoring rule satisfies the geometric property that $\E_{i\sim\alpha}[\rho(\beta)_i] = ||\alpha||\cos(\alpha, \beta)$, where $\cos(\alpha, \beta)$ is the cosine of the angle formed by vectors $\alpha, \beta$. Moreover, the spherical scoring rule is \emph{effective} with respect to the renormalized $L^2$ metric, i.e. $\E_{i\sim\alpha}[\rho(\beta)_i] > \E_{i\sim\alpha}[\rho(\beta')_i]$ if and only if $||\rho(\alpha) - \rho(\beta)|| < ||\rho(\alpha) - \rho(\beta')||$. Note that the spherical scoring rule plays the role of the assignment function $s$ in the more general preference framework.

We use the following straightforward observation bounding the deviation from the maximum possible payoff in terms of the renormalized $L^2$ distance from the true type. 

\begin{lemma}\label{lem:scoring_bound}
Let $||\rho(\alpha)-\rho(\alpha')||\le\lambda$. Then 
$$\E_{i\sim\alpha}[\rho(\alpha')_i]\ge\E_{i\sim\alpha}[\rho(\alpha)_i] - \frac{1}{2}\lambda^2.$$
\end{lemma}

\begin{proof}
We can write $||\rho(\alpha) - \rho(\alpha')||^2 = 2(1-\cos(\alpha, \alpha'))$, so $$
\E_{i\sim\alpha}[\rho(\alpha')_i] =  ||\alpha||\cos(\alpha, \alpha') 
= ||\alpha||\left(1-\frac{1}{2}||\rho(\alpha) - \rho(\alpha') ||^2\right) 
\ge\E_{i\sim\alpha}[\rho(\alpha)_i] - \frac{1}{2}\lambda^2.$$\end{proof}

\begin{thmn}[\ref{thm:eu}] The IC learning complexity of $\mathcal{A}$ is $O\left(n^{3/2}\log n\max\left(\log\frac{n}{\varepsilon}, \log\frac{1}{\tau}\right)\right)$.\end{thmn}

\begin{proof} 
Suppose we run $\mathcal{A}$ for $T$ rounds to achieve $(\varepsilon, \delta)$-learning. By Theorem~\ref{thm:eu_learning}, it holds with high probability that the hypothesis set obtained will be contained inside a small ball with respect to renormalized $L^2$ distance. More precisely, if $\mathcal{A}$ is given access to oracle $\hat{\alpha}$, and $\lambda = 2\varepsilon/n$, then $$\Pr[\mathcal{H}^T(\alpha)\subseteq B(\alpha, \lambda)] \ge 1-\delta,$$ where $\mathcal{H}^T(\alpha)$ is shorthand to denote a hypothesis set drawn from an execution of $\mathcal{A}_T(\hat{\alpha})$ and $B(\alpha, \lambda) = \{\alpha' : ||\rho(\alpha) - \rho(\alpha')||\le\lambda\}$.

By Lemma \ref{lem:scoring_bound}, $$\Pr\left[\forall\alpha^h\in\mathcal{H}^T(\alpha), \E_{i\sim\alpha}[\rho(\alpha^h)_i]\ge\E_{i\sim\alpha}[\rho(\alpha)_i]-\frac{1}{2}\lambda^2\right] \ge \Pr[\mathcal{H}^T(\alpha)\subseteq B(\alpha, \lambda)]\ge 1-\delta,$$
so any strategy can yield an advantage of at most $\frac{1}{2}\lambda^2 = \frac{2\varepsilon^2}{n^2}$ over truthful reporting. For $\varepsilon < n\sqrt{\tau}$ and $\delta = \nu$ we get $(\tau, \nu)$-incentive compatibility. The IC complexity is the query complexity of $(n\sqrt{\tau}, \nu)$-learning, which is $O(n^{3/2}\log n\log\frac{1}{\tau})$.

Thus, the number of queries required to simultaneously achieve $(\varepsilon, \delta)$-learning and $(\tau, \nu)$-incentive compatibility is $O\left(n^{3/2}\log n\max\left(\log\frac{n}{\varepsilon}, \log\frac{1}{\tau}\right)\right)$.\end{proof}

Our notion of incentive compatibility is approximate, and allows for small gains to the agent from misrepresenting their beliefs. We now demonstrate that, even though Theorem \ref{thm:eu} allows for the possibility of gaining some advantage by playing strategically, we can ensure that with high probability any type learned by the analyst as a result of a strategic interaction will be sufficiently close to the true type that the analyst accurately learns the true type nonetheless.

Suppose the analyst wants to achieve $\varepsilon$ learning accuracy, and additionally wants to guarantee that with probability at least $1-\delta$, any best-responding agent will report a belief, or type, that is within $\varepsilon$ total variation distance to the true type (note that this is a slightly different notion of incentive compatibility than the previous one). As usual, let $\alpha$ denote the agent's true type. Let $\varepsilon_0 < \frac{\varepsilon}{3}$, $\lambda = \frac{2\varepsilon_0}{n}$, $\delta_0 < 1-\sqrt{1-\delta}$, and run the algorithm to achieve $(\varepsilon_0, \delta_0)$-learning. 

We first show that any misreport that is sufficiently far from the true type yields, with high probability, a strictly worse payoff than truthful reporting. Recall that $B(\alpha, \lambda)$ denotes the closed ball of radius $\lambda$ centered at $\alpha$ with respect to the renormalized $L^2$ distance.

Suppose that $||\rho(\alpha) - \rho(\beta)|| > 2\lambda$, so that $B(\alpha, \lambda)\cap B(\beta, \lambda)=\emptyset$. Then, as the spherical scoring rule is effective with respect to renormalized $L^2$-distance, $$\begin{aligned}
\Pr &\left[\forall \alpha^h\in\mathcal{H}^T(\alpha), \forall \beta^h\in\mathcal{H}^T(\beta), \E_{i\sim\alpha}[\rho(\alpha^h)_i] > \E_{i\sim\alpha}[\rho(\beta^h)_i]\right] \\
&\ge \Pr[\mathcal{H}^T(\alpha)\subseteq B(\alpha, \lambda) \wedge \mathcal{H}^T(\beta)\subseteq B(\beta, \lambda)] \\ &\ge (1-\delta_0)^2 \\ &\ge 1-\delta,
\end{aligned}$$ so it holds with high probability that any such misreport yields a strictly worse payoff. 

The remaining misreports are sufficiently close to the true type such that the analyst does not care if these allow the agent to increase his payoff. Indeed, if $||\rho(\alpha)-\rho(\beta)||\le 2\lambda$, since $\mathcal{H}^T(\alpha)\subseteq B(\alpha, \lambda)$ and $\mathcal{H}^T(\beta)\subseteq B(\beta, \lambda)$ with high probability, the triangle inequality yields $$||\rho(\alpha) - \rho(\beta^h)||\le ||\rho(\alpha) - \rho(\beta)|| + ||\rho(\beta) - \rho(\beta^h)|| \le 3\lambda,$$ so $||\alpha-\beta^h||_{TV}\le 3\varepsilon_0 < \varepsilon$. Thus with probability at least $1-\delta$, the analyst will learn a $\beta^h$ such that $||\alpha-\beta^h||_{TV}\le\varepsilon$ for any such misreport.

To summarize, the algorithm can be run for $O(n^{3/2}\log n\log\frac{n}{\varepsilon})$ rounds (the exact number of rounds would just be a small constant factor more than that required by the vanilla learning requirement) such that regardless of what strategy an agent may use to best respond during the interaction, with high probability the analyst will end up accurately learning the agents true type. 

\section{Concluding remarks}
We have analyzed the incentive compatibility of active learning using data labeled by human subjects. Our results are directly applicable to the adaptive design of economic experiments that seek to estimate subjects' preference parameters. Our paper has discussed some of the leading areas of economic experimentation: estimation of risk aversion from multiple price lists, and belief elicitation using convex budgets and scoring rules. We highlight some challenges in making active learning compatible with incentives, but for the most part we offer satisfactory algorithmic solutions to the specific areas of experimentation we have focused on. 

There are, of course, many other areas of application of active learning, and incentive issues will be important as long as the required training data is labeled by human subjects under incentivized conditions. To this end, we have introduced a general model of active learning under incentives: we believe that we are the first to do so, and we expect our findings to motivate additional investigations of the problems at the intersection of learning and incentives.

\bibliographystyle{alpha}
\bibliography{arxiv-version}

\end{document}